\documentclass{article}
\usepackage[utf8]{inputenc}
\usepackage{amsmath,amssymb,mathtools}
\usepackage{amsthm}
\usepackage{xcolor}

\newtheorem{theorem}{Theorem}
\newtheorem{lemma}{Lemma}
\newtheorem{definition}{Definition}

\title{Computing all monomials of degree $n-1$ using $2n-3$ AND gates}
\author{Thomas H\"aner}
\date{\small Amazon Quantum Solutions Lab, Z\"urich, Switzerland}

\begin{document}

\maketitle

\begin{abstract}
We consider the vector-valued Boolean function $f:\{0,1\}^n\rightarrow \{0,1\}^n$ that outputs all $n$ monomials of degree $n-1$, i.e., $f_i(x)=\bigwedge_{j\neq i}x_j$, for $n\geq 3$. Boyar and Find have shown that the multiplicative complexity of this function is between $2n-3$ and $3n-6$. Determining its exact value has been an open problem that we address in this paper. We present an AND-optimal implementation of $f$ over the gate set $\{\text{AND},\text{XOR},\text{NOT}\}$, thus establishing that the multiplicative complexity of $f$ is exactly $2n-3$.
\end{abstract}

\section{Introduction}

The multiplicative complexity of a Boolean function $f:\{0,1\}^n\rightarrow \{0,1\}^m$ is the minimal number of AND gates required to implement $f$ over $\{\land, \oplus, 1\}$, where $\land$ is the logical AND of two Boolean inputs, $\oplus$ computes the exclusive OR of an arbitrary number of Boolean inputs, and the constant $1$ input can be used to invert a Boolean input $\overline x = x\oplus 1$. The multiplicative complexity is thus a good measure of the implementation cost of a function in cases where AND gates are much more costly than XOR gates. This is the case, for example, in fault-tolerant quantum computing \cite{meuli2019role} and secure computation protocols \cite{albrecht2015ciphers}.

While it is computationally intractable to compute the multiplicative complexity for a general function \cite{find2014complexity}, there are specific (classes of) functions for which the exact multiplicative complexity is known~\cite{boyar2008tight,ccalik2019multiplicative,turan2014multiplicative,ccalik2020boolean,haner2022multiplicative}.

Boyar and Find~\cite{boyar2018multiplicative} have shown that the vector-valued Boolean function $f(x)$ where each output $f_i(x)$ for $i\in\{1,...,n\}$ is given by
\begin{equation}\label{eq:monomial-function}
    f_i(x) = \bigwedge_{j\in\{1,...,n\}\setminus\{i\}} x_j
\end{equation}
has multiplicative complexity between $2n-3$ and $3n-6$. Boyar and Find prove the lower bound using an iterated algebraic degree argument, and they provide a construction that computes $f(x)$ with $3n-6$ AND gates.

\paragraph{Our contribution.} We improve upon the construction by Boyar and Find, and present an AND-optimal implementation of $f(x)$ using $2n-3$ AND gates, allowing us to conclude that the multiplicative complexity of $f(x)$ is exactly $2n-3$. This solves an open problem from Ref.~\cite{boyar2018multiplicative}.

\section{Preliminaries}

\begin{definition}[Algebraic Normal Form (ANF)]
For a Boolean function $f:\{0,1\}^n\rightarrow\{0,1\}$, its algebraic normal form is the unique representation
\[
    f(x)=\bigoplus_{I\subset\{1,...,n\}} a_I \bigwedge_{i\in I}x_i,
\]
with $a_I\in\{0,1\}$ and $x_i$ denoting the $i$th bit of the integer $x\in\{0,...,2^n-1\}$. Each $\bigwedge_{i\in I}x_i$ where $a_I=1$ is called a monomial of $f$.
\end{definition}

\begin{definition}[Algebraic Degree]
For a Boolean function $f:\{0,1\}^n\rightarrow\{0,1\}$, its algebraic degree, denoted by $\operatorname{deg}(f)$ is
\[
    \operatorname{deg}(f)=\max_{I\subset\{1,...,n\}} a_I |I|,
\]
where $a_I\in\{0,1\}$ denote the ANF coefficients of $f$ and $|I|$ is the number of elements in the set $I$.
\end{definition} 

\begin{definition}[Multiplicative Complexity]
For a Boolean function $f:\{0,1\}^n\rightarrow\{0,1\}$, its multiplicative complexity, denoted by $c_\land(f)$, is defined as the smallest number of AND gates in any implementation of $f$ consisting only of AND gates with two Boolean inputs, exclusive OR gates, and NOT gates.
\end{definition} 

One general way to obtain a lower bound on the multiplicative complexity of a function is by the degree lower bound~\cite{schnorr1988multiplicative}.
\begin{lemma}[Proposition 3.8 in \cite{schnorr1988multiplicative}]
For all Boolean functions $f$, it holds that $c_\land(f) \geq deg(f)-1$.
\end{lemma}

We use the short-hand notation $x_1\cdots x_n$ to represent $\bigwedge_{i=1}^n x_i$, and we refer to the computation of a logical AND of two Boolean values $x,y$, i.e., $x\land y$, as multiplication of $x$ by $y$. Similarly, we refer to the computation of a logical exclusive OR (XOR) of $x$ and $y$, i.e., $x\oplus y$, as addition (modulo 2).

We say a monomial has a ``gap'' at $x_k$ if the monomial is of the form $\bigwedge_{i\in I}x_i$ and $k\notin I$. When writing down monomials explicitly as $x_ix_j\cdots x_m$, we assume that the variables have been ordered such that $i<j<\cdots <m$. Similarly, we say that a multiplication of a monomial $x_ix_j\cdots x_m$ by $x_n$ appends $x_n$ to the monomial if $m<n$.

\section{Construction}\label{sec:construction}

In this section, we present an AND-optimal construction to evaluate $f(x)$ using $2n-3$ AND gates.

We evaluate $f(x)$ in 3 stages. In the first stage, we compute the XOR of all monomials of degree $n-1$, i.e., for $n$ inputs $x_1,...,x_n$ the output of the first stage is
\begin{equation}\label{eq:stage1}
   s_0^n := \bigoplus_{i=1}^n \bigwedge_{j\neq i} x_j.
\end{equation}
Boyar and Peralta~\cite{boyar2008tight} have shown that $s_0^n$ may be computed using $n-2$ AND gates, and that this is optimal.

The second stage produces an additional $n-1$ intermediate outputs starting from $s_0^n$. Specifically, if $n$ is odd, then each of these additional outputs is the XOR of two monomials from $s_0^n$ such that all $n$ intermediate outputs are linearly independent, i.e., all monomials can be extracted from these $n$ intermediate outputs using XORs. If $n$ is even, then the same is true for the first $n-2$ outputs of the second stage, but the final output is just the monomial $x_1,...,x_{n-1}$. In both cases, the number of AND gates used by the second stage is $n-1$.

The third and final stage combines the $n-1$ outputs of the second stage with $s_0^n$ using XORs in order to generate the $n$ different monomials corresponding to the $n$ outputs of $f(x)$. The total number of AND gates used to evaluate $f(x)$ is then $(n-2)+(n-1)=2n-3$.

\subsection{Stage 1}

An AND-optimal construction for computing $s_0^n$ with $n-2$ AND gates was found by Boyar and Peralta~\cite[Lemmas 12 and 13]{boyar2008tight}:

\begin{lemma}[Special case of Lemma 12 in \cite{boyar2008tight}]\label{lm:evenmonomials}
Let the number of inputs $n$ be even. Then, $s_0^n$ can be computed from $s_0^{n-1}$ via
\[
    s_0^n = s_0^{n-1}\land \bigoplus_{i=1}^n x_i.
\]
\end{lemma}

\begin{lemma}[Lemma 13 in \cite{boyar2008tight}]\label{lm:oddmonomials}
For odd $n$, $s_0^n$ may be computed using the recursion
\[
    s_0^n = s_0^{n-2}\land (((x_{n-1}\oplus x_n)\land \bigoplus_{i=1}^{n-1} x_i) \oplus x_{n-1})
\]
and the base case $s_0^3=((x_1\oplus x_2)\land (x_2\oplus x_3))\oplus x_2=x_1x_2\oplus x_2x_3\oplus x_1x_3$.
The multiplicative complexity of $s_0^n$ is $n-2$.
\end{lemma}

We note that, if $n$ is even, this construction computes $s_0^{n-1}$ using the recursion in Lemma~\ref{lm:oddmonomials} and then uses Lemma~\ref{lm:evenmonomials} to arrive at $s_0^n$. Our second stage will make use of the intermediate result $s_0^{n-1}$ if $n$ is even.

\subsection{Stage 2}

In the second stage of our construction, we generate $n-1$ additional linearly independent intermediate results that can be used to extract all $n$ monomials using XORs in the third stage. Specifically, for odd $n$, we compute
\[
    s_i^n := (x_i\oplus x_{i+1})\land s_0^n
\]
for all $i\in\{1,...,n-1\}$. For even $n$, we compute the same $s_i^n$ for $i\in\{1,...,n-2\}$ and we additionally compute the last output of $f(x)$ directly via 
\begin{equation}\label{eq:evenstage2}
    f_n(x)=s_0^{n-1}\land\bigoplus_{i=1}^{n-1}x_i.
\end{equation}
We claim that the ANF of $s_i^n$ contains exactly those two monomials of degree $n-1$ where either $x_i$ or $x_{i+1}$ is missing. We prove this next, before proving the equality in \eqref{eq:evenstage2}.

\begin{lemma}\label{lm:lin-indep-anfs}
Let $n$ be the number of inputs. Then, for each $i\in\{1,...,n-1\}$, the ANF of
\[
    s_i^n := (x_i\oplus x_{i+1})\land s_0^n
\]
consists of exactly two monomials of degree $n-1$; one where $x_i$ is missing, and one where $x_{i+1}$ is missing.
\end{lemma}
\begin{proof}
Each of the $n$ monomials of degree $n-1$ in the ANF of $s_0^n$ is multiplied by $(x_i\oplus x_{i+1})$. For each monomial where both $x_i$ and $x_{i+1}$ are present, a multiplication by either of these two variables results in the same monomial, and they thus cancel. Multiplying the two monomials where $x_i$ or $x_{i+1}$ is missing by $(x_i\oplus x_{i+1})$ results in (1) the degree-$n$ monomial $x_1\cdots x_n$, and (2) the monomial itself (where $x_i$ or $x_{i+1}$ is missing) for each of the two monomials. The degree-$n$ monomial is thus generated twice and, therefore, the only two monomials remaining in the ANF of the multiplication result are those two where $x_i$ or $x_{i+1}$ is missing.
\end{proof}

Finally, we show that the equality in \eqref{eq:evenstage2} holds for even $n$.
\begin{lemma}\label{lm:even-stage2}
Let the number of inputs $n$ be even. Then the last output of $f(x)$ may be computed from $s_0^{n-1}$ using one additional AND gate via
\[
    f_n(x)=s_0^{n-1}\land\bigoplus_{i=1}^{n-1}x_i.
\]
\end{lemma}
\begin{proof}
The $n-1$ monomials in the ANF of $s_0^{n-1}$ are of even degree $n-2$ with variables from $\{x_1,...,x_{n-1}\}$. Therefore, for each monomial, $n-2$ out of the $n-1$ variables in $\{x_1,...,x_{n-1}\}$ are already present in the monomial, and a multiplication by such an $x_i$ results in the same monomial. As this happens an even number of times, all of these terms cancel. For each monomial, the only nontrivial contribution to the result comes from the $x_i$ not present in the monomial, and the contribution is the same for each monomial, namely $x_1\cdots x_{n-1}$. Because this contribution is added to the result an odd number of times (once for each monomial in $s_0^{n-1}$), the only monomial that is left in the ANF of the multiplication result is $x_1\cdots x_{n-1}$, which is equal to $f_n(x)$, as claimed.
\end{proof}

\subsection{Stage 3}
After completing stages 1 and 2, we have $n$ linearly independent intermediate results, each containing either 1 (if $n$ is even), 2, or $n$ monomials of degree $n-1$. In this final stage, we combine these $n$ intermediate results using XORs in order to compute the outputs of $f(x)$. This allows us to prove our main result.

\begin{theorem}
The multiplicative complexity of the vector-valued Boolean function $f$ with $n$ Boolean inputs $x_1,...,x_n$, where the $i$-th output is given by
\[
f_i(x_1,...,x_n)=\bigwedge_{j\neq i} x_j,
\]
for $i\in\{1,...,n\}$ is $c_{\land}(f)=2n-3$.
\end{theorem}
\begin{proof}
Since Boyar and Find have shown that $c_{\land}(f)\geq 2n-3$~\cite{boyar2018multiplicative}, it remains to show that $c_\land(f)\leq 2n-3$ by completing our construction, which establishes that $c_\land(f)=2n-3$.

We first consider the case where $n$ is odd.
Lemma~\ref{lm:oddmonomials} shows that $s_0^n$ can be computed using $n-2$ AND gates. From Lemma~\ref{lm:lin-indep-anfs} we know that each $s_k^n$ for $k\in\{1,...,n-1\}$ can be computed from $s_0^n$ using a single AND gate.
We may compute the first output $f_1(x)$ via
\[
    f_1(x)=s_0^n\oplus \bigoplus_{i=1}^{\frac{n-1}2} s_{2i}^n .
\]
To see that this is correct, note that $\bigoplus_{i=1}^{\frac{n-1}2} s_{2i}^n$ contains all degree-$(n-1)$ monomials except $f_1(x)=x_2\cdots x_n$, whereas $s_0^n$ contains all degree-$(n-1)$ monomials.
This first output $f_1(x)$ may now be used to obtain $f_2(x)$ via $f_1(x)\oplus s_1^n$. In turn, $f_3(x)=f_2(x)\oplus s_2^n$, and so on, until the final output $f_n(x)=f_{n-1}(x) \oplus s_{n-1}^n$ has been computed. The total number of AND gates in this case is $(n-2) + (n-1)=2n-3$.

For even $n$, Lemma~\ref{lm:oddmonomials} shows that $s_0^{n-1}$ can be computed using $n-3$ AND gates. From Lemma~\ref{lm:evenmonomials}, $s_0^n$ can be computed from $s_0^{n-1}$ using a single AND gate, and Lemma~\ref{lm:even-stage2} shows that a single AND gate is sufficient to compute $f_n(x)$ from $s_0^{n-1}$.
In addition, we compute $s_k^n$ for $k\in\{1,...,n-2\}$ using $n-2$ AND gates. The total AND gate count is thus $(n-3)+1+1+(n-2)=2n-3$, and we can extract the first output as follows:
\[
    f_1(x)=s_0^n\oplus f_n(x) \oplus \bigoplus_{i=1}^{\frac{n-2}2} s_{2i}^n .
\]
To see that this is correct, note that $\bigoplus_{i=1}^{\frac{n-2}2} s_{2i}^n$ contains all degree-$(n-1)$ monomials except $f_1(x)=x_2\cdots x_n$ and $f_n(x)=x_1\cdots x_{n-1}$, whereas $s_0^n$ contains all degree-$(n-1)$ monomials.
To compute the remaining outputs $f_i(x)$ for $i\in\{2,...,n-1\}$, we may again use that $f_i(x)=f_{i-1}(x)\oplus s_{i-1}^n$.
\end{proof}

\bibliographystyle{unsrt}
\bibliography{references}

\end{document}